\documentclass[oneside]{article}

\usepackage[utf8]{inputenc}
\usepackage[a4paper, margin=4cm]{geometry}
\usepackage{parskip}
\usepackage{framed}
\usepackage{comment}
\usepackage{amsmath} 
\usepackage{dsfont}
\usepackage{amssymb}
\usepackage{enumerate}
\usepackage[english]{babel}
\usepackage{tipa}
\usepackage{amsmath}
\usepackage{wasysym}
\usepackage{makecell}
\usepackage{listings}
\usepackage{url}
\usepackage{mathtools}
\usepackage{amsthm}
\usepackage[colorlinks=true,
            linkcolor=blue,
            citecolor=blue,
            urlcolor=blue]{hyperref}
\usepackage{cleveref}
\usepackage{multirow}
\usepackage{subfig,graphicx}
\usepackage[toc,page,title]{appendix}
\PassOptionsToPackage{table}{xcolor}
\usepackage{colortbl}
\usepackage{booktabs}
\usepackage[table]{xcolor}
\usepackage{array}
\usepackage{caption}
\usepackage[export]{adjustbox}
\usepackage[numbers,sort&compress]{natbib}
\usepackage{float}
\floatstyle{plaintop}
\restylefloat{table}

\numberwithin{equation}{section}
\newtheorem{theorem}{Theorem}[section] 
\newtheorem{lemma}[theorem]{Lemma} 
\newtheorem{example}[theorem]{Example}

\begin{document}

\title{The role of root allometry and root overlap in vegetation pattern formation
\author{
Jelle van der Voort\thanks{Mathematical Institute, Leiden University, 2300 RA, Leiden, Netherlands} 
\quad 
Ricardo Martinez-Garcia\thanks{\parbox[t]{0.85\textwidth}{Center for Advanced Systems Understanding,(CASUS) – Helmholtz-Zentrum Dresden-Rossendorf (HZDR), Untermarkt 20, Görlitz, 02826, Germany}}\space\thanks{\parbox[t]{0.85\textwidth}{ICTP South American Institute for Fundamental Research \& Instituto de Física Teórica, Universidade Estadual Paulista- UNESP, R. Dr. Bento Teobaldo Ferraz, 271- 2- Várzea da Barra Funda, São Paulo- SP, 01140-070, Brazil}} \space\thanks{\parbox[t]{0.85\textwidth}{Department of Ecology, Institute of Biosciences, University of São Paulo, São Paulo, Brazil}} 
\quad 
Arjen Doelman$^*$
}
}

\maketitle

\begin{abstract}

\noindent We investigate how root allometry and below-ground root competition due to root overlap affect the conditions for pattern formation in non-local scalar vegetation models. We start from a spatially explicit two-component vegetation-water system in which plants acquire soil water through their laterally extended root systems, with root allometry represented by an interaction kernel whose spatial extent increases with the above-ground biomass density. In the limit of fast water dynamics, we eliminate the water variable, resulting in a non-local scalar model for the vegetation dynamics. This reduction reveals that accounting for root overlap naturally leads to a double integral structure in the scalar model, in contrast to the single convolution term commonly used to model non-local competition for resources. For the resulting scalar model, we identify two distinct mechanisms that can drive Turing instabilities. The first instability mechanism requires root growth to outpace mortality, while the second mechanism relies on the development of competitive pressure zones, in which vegetation growth is impeded by the cumulative competitive pressure exerted by neighbouring areas with higher biomass density. Root allometry promotes the occurrence of both instability mechanisms, while root overlap inhibits the emergence of competitive pressure zones. Hence, our results show that both root allometry and root overlap are important features to consider when studying vegetation pattern formation.

\end{abstract}

\vspace*{\fill}

Author for correspondence: Jelle van der Voort\\
Email: \textit{j.van.der.voort@math.leidenuniv.nl} 

\newpage

\section{Introduction}

Over the last decades, considerable research effort has been devoted to establishing the conditions under which vegetation pattern formation can occur \citep{Rietkerk2008,Borgogno2009,Meron2019,Gandhi2019}. Analyses of scalar PDE models have consistently shown that non-local competitive interactions can induce Turing instabilities that drive patterning \citep{Britton1989,Lefever1997,Fuentes2003,Clerc2005,Leimar2008,Martinez2013a,Martinez2013b,Martinez2014,Tega2022,vanderVoort2026}. In such models, non-local competition is typically represented by a single convolution term where the interaction kernel determines how the dynamics at a given location is influenced by its surroundings. However, in the context of vegetation dynamics, where non-local competition typically stems from root-mediated competition for water, this approach has two important limitations.

The first drawback is the simplifying assumption that the interaction range of non-local competition is fixed. This assumption neglects life-stage-dependent effects, such as the much smaller root extents of seedlings and saplings as compared to mature plants. A more suitable modelling framework should account for plants extending their root systems as they grow. This way, plants gain access to previously inaccessible water sources as they mature, thereby increasing their total water uptake and stimulating further growth. Several modelling studies have included this positive feedback mechanism and shown that it promotes the formation of spatial patterns in vegetation-water models \citep{Gilad2007,Bennett2023} and scalar vegetation models \citep{Lefever2009,Cisternas2020,Tlidi2024}. Nevertheless, a systematic derivation of how root allometry affects the onset of patterning through a Turing bifurcation has not yet been carried out.

The second shortcoming concerns below-ground root competition when the roots of different plants overlap, which we refer to as root-to-root competition. This form of competition is not captured by scalar models in which non-local competition is represented through a single convolution term. Such an integral accounts for the non-local water uptake through a plant's own root system and the resulting competitive pressure it exerts on above-ground biomass in its surroundings. However, it neglects that neighbouring plants may also exploit the same water sources when their root systems partly overlap. Consequently, when two plants extract water from the same source, yet neither root system extends beyond the above-ground biomass location of the other plant, their competition for that shared source is not taken into account \citep{Maciel2021}. The cumulative water uptake is therefore overestimated, resulting in an underestimation of the competitive pressure between neighbouring plants. Overlap structures of this kind are well established in niche theory \citep{Macarthur1967,Schoener1974,HernandezGarcia2009}, but have not been considered in models of vegetation pattern formation.

The goal of this study is to resolve these two limitations and investigate how the inclusion of root allometry and proper modelling of root-to-root competition affect the conditions for pattern formation in non-local scalar vegetation models. 

This paper is structured as follows. In Section \ref{section: derivation of a scalar equation}, we introduce a two-component vegetation-water system that explicitly accounts for root-mediated water uptake, including root allometry. We then eliminate the explicit water variable by considering the limit of fast water dynamics, thereby obtaining a reduced scalar vegetation model of the full two-component system. In Section \ref{section: derivation of dispersion relation}, we derive the dispersion relation corresponding to the positive homogeneous equilibria of this scalar model, which can be used to determine whether Turing instabilities can occur\footnote{Readers who are primarily interested in the ecological implications rather than the technical derivation of the dispersion relation can skip Section \ref{section: derivation of dispersion relation} and proceed to Section \ref{section: pattern formation conditions}.}. In Section \ref{section: pattern formation conditions}, we derive and interpret the conditions for pattern formation, and analyze how root allometry and root overlap affect the resulting pattern-generating mechanisms. 

\section{Derivation of a scalar equation}
\label{section: derivation of a scalar equation}

In this section, we derive a scalar vegetation model that accounts for both root allometry and root-to-root competition. We first introduce a two-component vegetation-water model that includes non-local water acquisition by vegetation through their laterally extended root systems. By considering the limit of fast water dynamics, we obtain a reduced scalar model. Finally, we interpret the novel double integral structure that arises in this model and discuss the representation of both root allometry and root overlap.

\subsection{A vegetation-water system with non-local water uptake}
\label{section: model equations}

We consider the following two-component framework,
\begin{align}
\label{model: two component model with non local water uptake}
\begin{split}
    \dot{u} &= \overbrace{p_1s(u)\int \phi\left(\frac{x'-x}{r(u(x,t))}\right)w(x',t)\,dx'}^{\text{plant growth through water uptake}}\quad-\quad \overbrace{m(u)}^{\text{mortality}}\quad+\quad\overbrace{D_u\Delta  u}^{\text{biomass spread}},\\
    \dot{w} &= \underbrace{p_2}_\text{precipitation}\quad-\quad \underbrace{p_3w}_{\text{evaporation}}\quad-\quad \underbrace{p_4w\int \phi\left(\frac{x'-x}{r(u(x',t))}\right)s(u(x',t))\,dx'}_{\text{water uptake by plants}}.
    \end{split}
\end{align}
Here, $u(x,t)$ is the above-ground biomass density and $w(x,t)$ the soil water concentration, both at location $x$ and time $t$. The parameters $p_i$ represent the rates of biomass production through water uptake ($p_1$), precipitation ($p_2$), evaporation ($p_3$) and water uptake by vegetation ($p_4$). The continuous function $m(u) \geq 0$ (with $m(u) = 0$ only if $u=0$) models the mortality of vegetation. We retain generality by leaving it unspecified instead of, e.g., assuming a linear form for simplicity, as in \citep{Klausmeier1999,HilleRisLambers2001,Gilad2004}. We do assume $\lim_{u \rightarrow \infty} m(u)= \infty$, such that mortality increasingly limits plant growth at high biomass densities. Biomass spread is modelled as diffusion with diffusion coefficient $D_u$. In vegetation-water systems, slow biomass spread combined with fast soil water diffusion can typically generate diffusion-driven instabilities that lead to pattern formation \citep{Klausmeier1999,HilleRisLambers2001,Gilad2004}. Since our focus is on vegetation pattern formation generated by lateral root systems alone, we omit soil water diffusion. Consequently, our model is most directly applicable to arid regions in which lateral soil water transport is weak, for instance, in soils with low hydraulic conductivity \citep{Getzin2016}.

The kernel function $\phi(x)\geq 0$ describes the spatial distribution of roots around a plant. Assuming a flat terrain and isotropic root growth, the kernel should be symmetric, and thus $\phi(x) = \phi(-x)$. We further assume that the root density decreases monotonically with the distance from the focal plant, implying $\phi'(x)\leq 0$ for $x \geq 0$. Lastly, we require the kernel function to be normalized, i.e., $\int\phi(x)\,dx=1$, such that the model retains the same parameter values for spatially homogeneous solutions as in the spatially explicit setting. This property can be imposed on any ecologically reasonable competition kernel by scaling it with a multiplicative constant.

The functions $s(u)$ and $r(u)$ influence the spatial root distribution differently. As a measure of the total root density, $s(u)$ uniformly scales the magnitude of the root distribution while leaving its spatial extent unchanged (left panel of Figure \ref{figure: influence of r and s}). In contrast, $r(u)$ stretches or contracts the root distribution over space, thereby modifying the root extent (middle panel of Figure \ref{figure: influence of r and s}). Therefore, $s(u)$ and $r(u)$ jointly determine both the intensity and the spatial extent of non-local water uptake (right panel of Figure \ref{figure: influence of r and s}).

In the absence of vegetation, both effects should vanish, so $s(0)=r(0)=0$. Since roots supply plants with resources that are essential for survival, the total root density and root extent should be positive whenever vegetation is present. Therefore, we require $s(u),r(u)>0$ whenever $u>0$. Furthermore, we associate increases in biomass density with the maturation of vegetation, which includes the development of a denser and more extensive root system. Therefore, both functions should increase with $u$, leading to the conditions $s'(u)>0$ and $r'(u)>0$. Particularly, this last condition ensures that root allometry is included in the model: the root extent of vegetation increases with the above-ground biomass.

\begin{figure}[t]
    \centering
    \begin{minipage}[b]{0.32\textwidth}
        \centering
        \includegraphics[width=\textwidth]{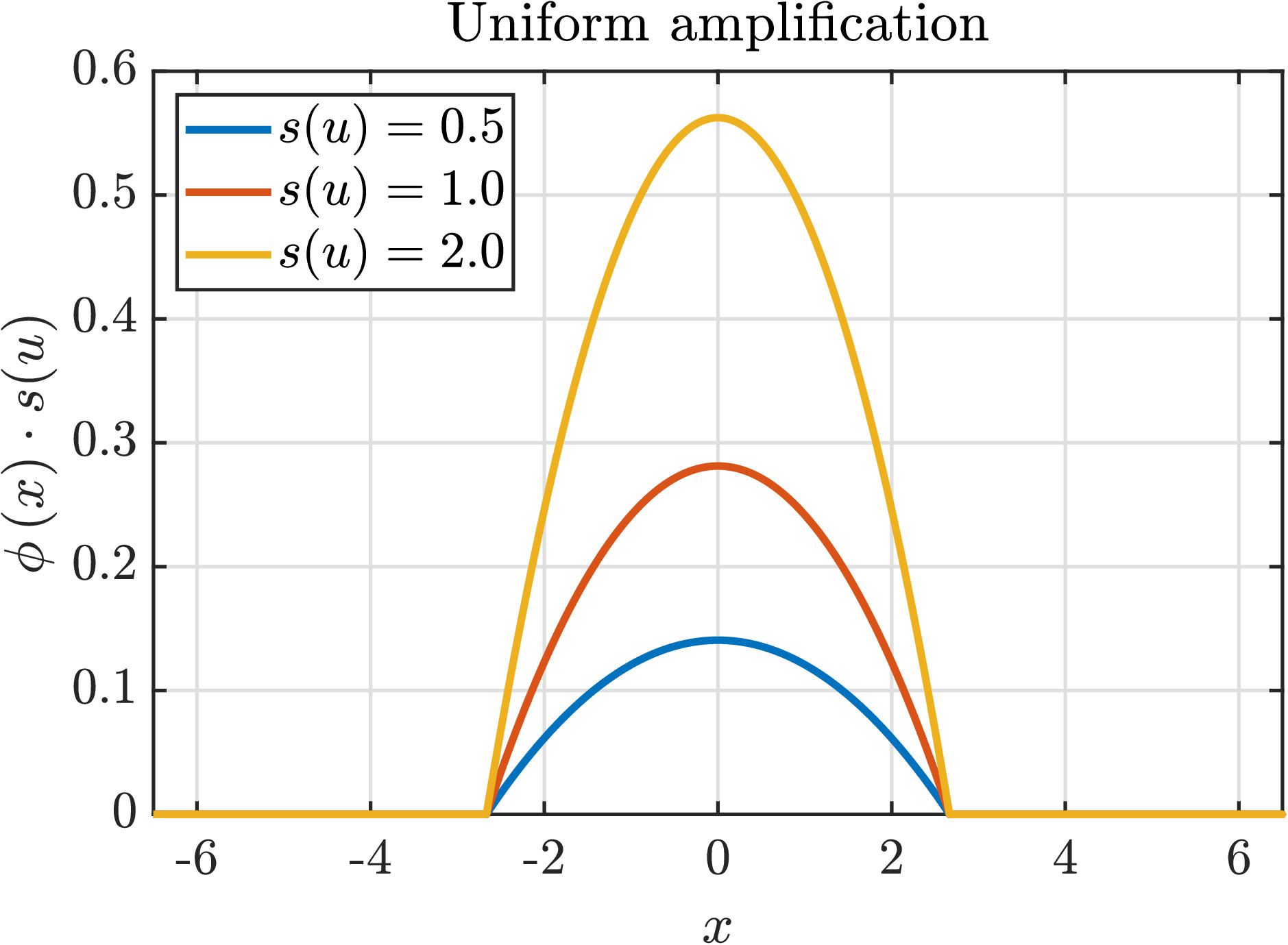}
    \end{minipage}%
        \hspace{0.1cm}
    \begin{minipage}[b]{0.32\textwidth}
        \centering
        \includegraphics[width=\textwidth]{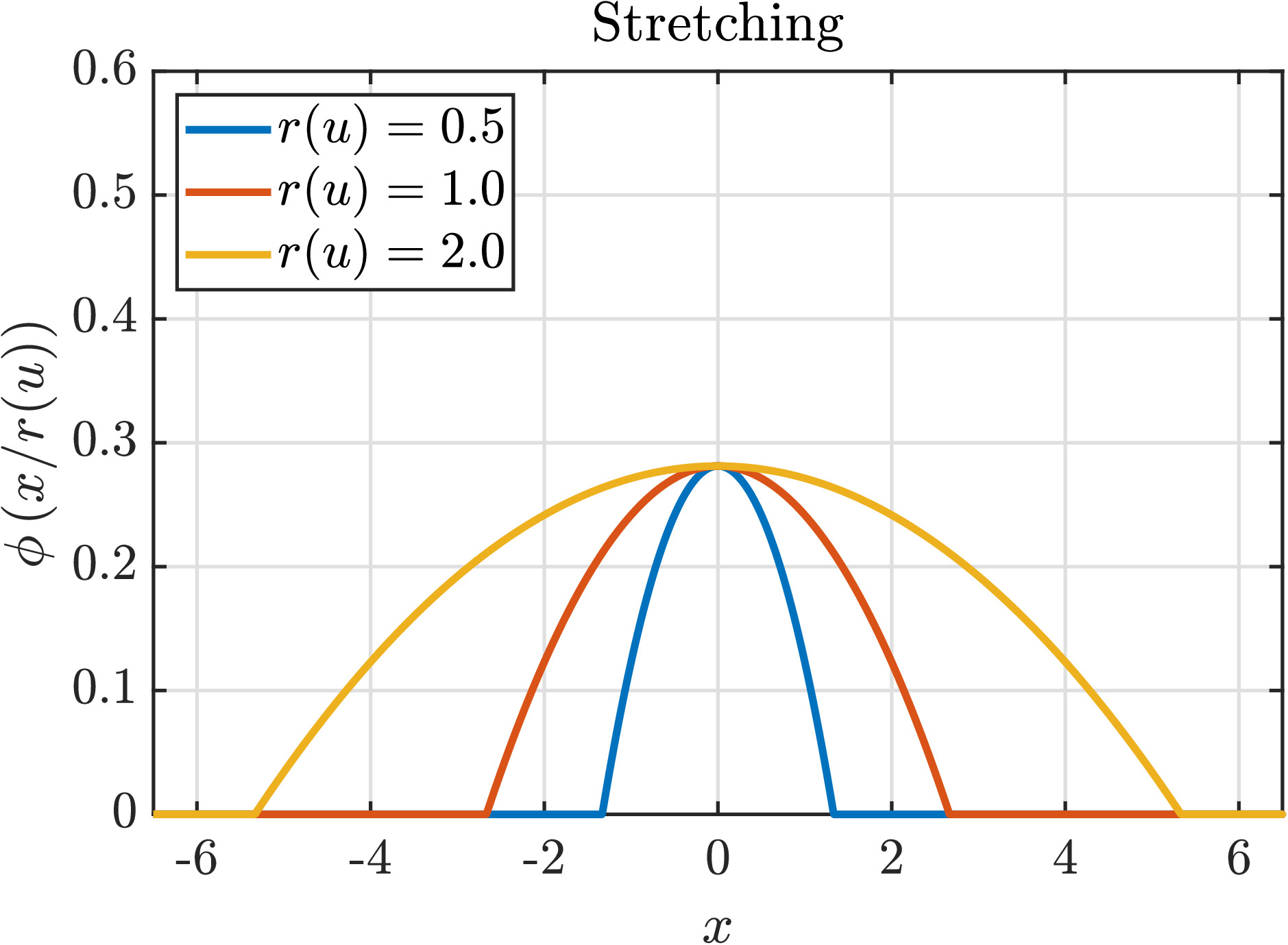}
    \end{minipage}%
    \hspace{0.1cm}
    \begin{minipage}[b]{0.32\textwidth}
        \centering
        \includegraphics[width=\textwidth]{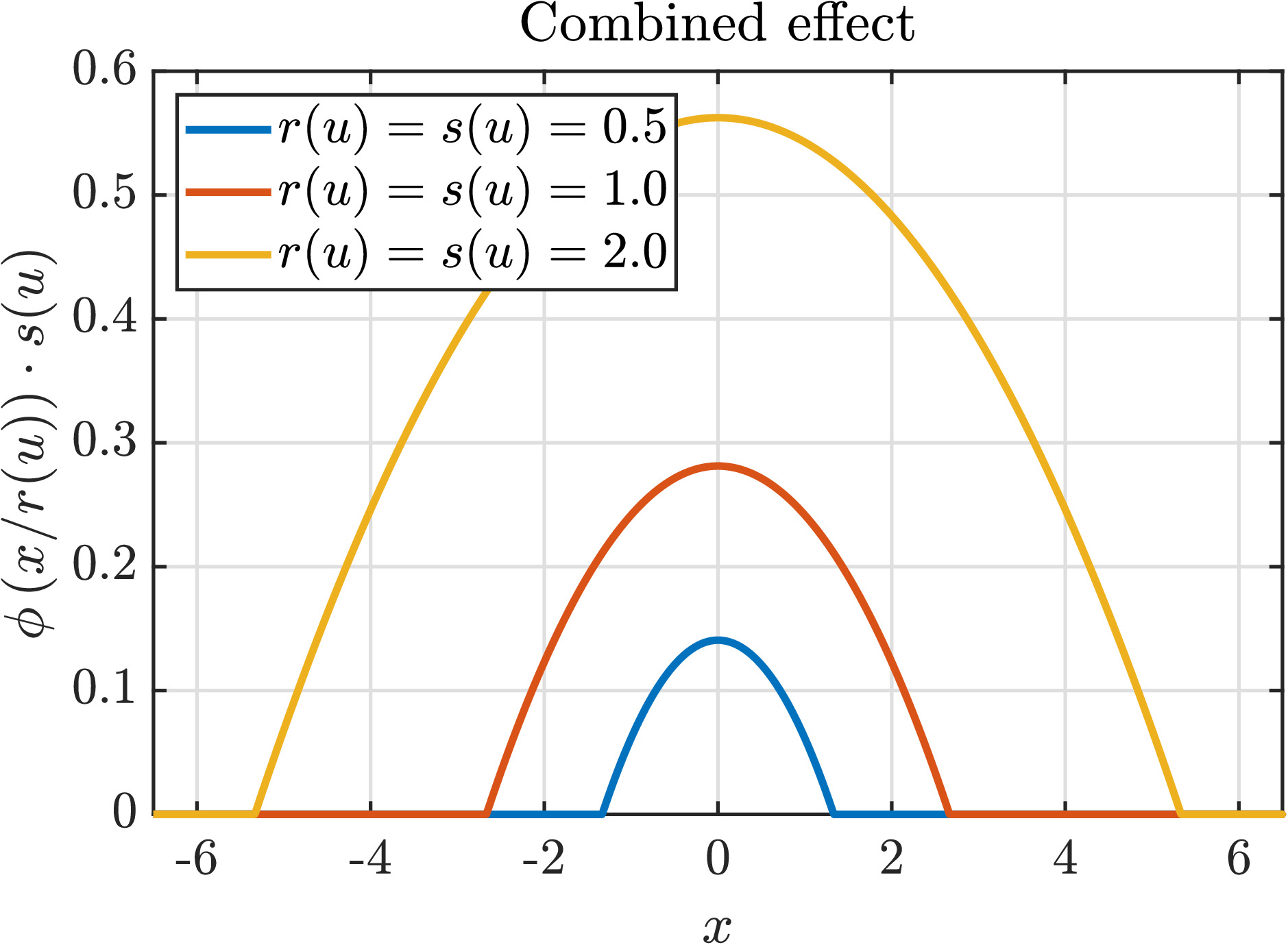}
    \end{minipage}%
\caption{The distinct roles of $s(u)$ and $r(u)$ in shaping the spatial root distribution $\phi\left(\frac{x}{r(u)}\right)s(u)$. Increasing $s(u)$ enhances the magnitude of the distribution uniformly without changing its spatial extent (left panel). In contrast, increasing $r(u)$ stretches the distribution over space, thereby increasing the root extent (middle panel). The combined effect of both mechanisms is shown in the right panel. We used a parabolic kernel function $\phi(x)$ to make the plots.}
\label{figure: influence of r and s}
\end{figure}

Finally, we note that $r$ depends on $u(x,t)$ in the vegetation equation, but on $u(x',t)$ in the water equation. This important distinction reflects the different reference points of the two processes. The vegetation equation describes the growth of a plant located at $x$ through the uptake of soil water in its neighbourhood, weighted by its own root system \citep{Meron2019}. Hence, the relevant root characteristics are those associated with $u(x,t)$. In contrast, the water equation describes the depletion of soil water at the location $x$ by all roots present at that location. These roots originate from vegetation at different locations $x'$ in the neighbourhood of $x$, so the integral over $x'$ accounts for the root characteristics associated with $u(x',t)$. This formulation preserves the overall balance between water uptake and depletion (see Appendix \ref{appendix: balance}). In particular, the water equation accounts for the depletion of a shared water source by roots originating from different plants, which in turn affects their above-ground growth through the vegetation equation. Hence, below-ground competition when root systems overlap is accurately represented.

\subsection{Fast water dynamics}

We can eliminate the water variable from (\ref{model: two component model with non local water uptake}) by assuming that all processes governing the soil water concentration (i.e., precipitation, evaporation and water uptake by plants) occur on a much faster time scale than that of the vegetation dynamics (i.e., biomass growth and mortality) \citep{HernandezGarcia2009,Martinez2014}. This separation of time scales can be incorporated by scaling the parameters associated with the water dynamics. Specifically, we introduce the small parameter $0< \delta \ll 1$ and set $p_2 = \frac{\tilde{p_2}}{\delta}$, $p_3 = \frac{\tilde{p_3}}{\delta}$ and $p_4 = \frac{\tilde{p_4}}{\delta}$, where $\tilde{p}_2$, $\tilde{p}_3$ and $\tilde{p}_4$ are of $\mathcal{O}(1)$. The remaining parameters and variables are assumed to be of $\mathcal{O}(1)$ as well. These scalings motivate an asymptotic expansion for the soil water variable of the form $w(x,t) = w_0(x,t) + \mathcal{O}(\delta)$. Substituting the scalings and this expansion into the $\dot{w}$-equation of (\ref{model: two component model with non local water uptake}) and collecting the leading-order terms gives,
\begin{align}
    0 = \tilde{p}_2-\tilde{p}_3w_0(x,t)-\tilde{p}_4w_0(x,t)\int \phi\left(\frac{x'-x}{r(u(x',t))}\right)s(u(x',t))\,dx'.
\end{align}
This equation yields a closed-form expression for the leading-order soil water concentration:
\begin{align}
\label{expression: approximation}
    w_0(x,t) = \frac{\tilde{p}_2}{\tilde{p}_3+\tilde{p}_4\int \phi\left(\frac{x'-x}{r(u(x',t))}\right)s(u(x',t))\,dx'}.
\end{align}
Since $\delta$ is small by assumption, the soil water concentration quickly approaches the quasi-steady state $w_0(x,t)$, as determined by the spatial distribution of vegetation. Substituting the approximation (\ref{expression: approximation}) for $w(x,t)$ directly into the $\dot{u}$-equation of (\ref{model: two component model with non local water uptake}) yields the scalar vegetation model that is the main focus of this paper:
\begin{align}
\begin{split}
\label{model: reduced nonlocal scalar equation}
    \dot{u} = \,&p_1s(u)\int \phi\left(\frac{x'-x}{r(u(x,t))}\right)\theta \left(\int \phi\left(\frac{x''-x'}{r(u(x'',t))}\right)s(u(x'',t))\,dx''\right)dx'\\&\quad\quad\quad\quad\quad\quad\quad\quad\quad\quad\quad\quad\quad\quad\quad\quad\quad\quad\quad\quad\quad\quad\quad\quad-m(u)+D_u \Delta u,
\end{split}
\end{align}
where, for notational convenience, we introduced,
\begin{align}
\label{definition: theta}
    \theta\left(y\right) \coloneqq \frac{\tilde{p}_2}{\tilde{p}_3+\tilde{p}_4 y} = \frac{p_2}{p_3+p_4 y}.
\end{align}

\subsection{Interpretation of the scalar model}
\label{section: double integral}

The reduced scalar model (\ref{model: reduced nonlocal scalar equation}) consists of three distinct contributions: root-mediated vegetation growth, vegetation mortality and biomass spread. Both vegetation mortality and biomass spread retain the same interpretation as in the two-component vegetation-water model (\ref{model: two component model with non local water uptake}), since neither is affected by the reduction procedure. The distinctive feature of (\ref{model: reduced nonlocal scalar equation}) is its double integral term, which represents vegetation growth through soil water acquisition via laterally extended roots. We explain the meaning of this term working from the inside outward.

We start with the argument of $\theta$, i.e., the integral that also appears in the $\dot{w}$-equation of (\ref{model: two component model with non local water uptake}). This integral measures the cumulative root density at the location $x'$ originating from plants in its surrounding area. The contribution to this quantity originating from above-ground biomass at $x''$ is weighted both by its total root density and its root extent through the functions $s(u(x'',t))$ and $r(u(x'',t))$, respectively. The higher the cumulative root density, the greater the competitive pressure for soil water. Therefore, this integral effectively measures the competitive pressure at the location $x'$.

Next, the function $\theta$ maps this competitive pressure at $x'$ to the available soil water at $x'$. Since higher competitive pressure corresponds to a lower soil water availability, $\theta$ should be decreasing as a function of the competitive pressure. This is indeed the case, as can be deduced from (\ref{definition: theta}).

Finally, the outer integral plays the same role as the integral in the original $\dot{u}$-equation of (\ref{model: two component model with non local water uptake}): it measures the non-local soil water uptake of the vegetation located at $x$ through its lateral root system, which is characterized by the root functions $s(u(x,t))$ and $r(u(x,t))$. The parameter $p_1$ then converts this total water uptake into biomass accumulation.

Root allometry is present in the scalar model (\ref{model: reduced nonlocal scalar equation}) through the function $r(u)$. It appears in the argument of the kernel functions and quantifies the spatial root extent as a function of the above-ground biomass density, analogously to the two-component setting (\ref{model: two component model with non local water uptake}) and other scalar vegetation models in the literature \citep{Lefever2009,Cisternas2020,Tlidi2024}.

Root-to-root competition, on the other hand, is included through the double integral structure. This distinguishes our model from scalar models in the literature that describe non-local resource competition using only a single convolution \citep{Britton1989,Lefever1997,Fuentes2003,Clerc2005,Leimar2008,Martinez2013a,Martinez2013b,Martinez2014,Tega2022,vanderVoort2026}. These models include the competitive pressure exerted by the root systems of vegetation, which directly affects above-ground biomass growth. In our model, however, we also account for the fact that this growth depends on an associated root system as well, which supplies the soil water necessary for biomass accumulation. The reduction procedure shows that properly representing this dynamic (i.e., soil water uptake when root systems overlap) requires two interacting integrals rather than a single one.

\section{Derivation of the dispersion relation}
\label{section: derivation of dispersion relation}

In this section, we derive the dispersion relation associated with the positive homogeneous equilibria of the scalar vegetation model (\ref{model: reduced nonlocal scalar equation}). First, we discuss the existence of homogeneous equilibria. Then, we perform linear stability analysis around an arbitrary positive homogeneous equilibrium, yielding a first form of the dispersion relation. Lastly, we simplify its terms to obtain a more tractable expression.

\subsection{Homogeneous equilibria}
\label{subsection: homogeneous equilibria}

A homogeneous equilibrium is a steady-state solution that is also constant in space, i.e., $\bar{u}(x,t) = \bar{u}$. Substituting this expression into (\ref{model: reduced nonlocal scalar equation}) yields,
\begin{align}
    m(\bar{u}) = p_1s(\bar{u}) \int \phi\left(\frac{x'-x}{r(\bar{u})}\right)\theta\left(\int \phi\left(\frac{x''-x'}{r(\bar{u})}\right)s(\bar{u})dx''\right) dx'.
\end{align}
By introducing the change of integration variable $y = \frac{x''-x'}{r(\bar{u})}$, the inner integral can be calculated as,
\begin{align}
    \int \phi\left(\frac{x''-x'}{r(\bar{u})}\right)s(\bar{u})dx'' = s(\bar{u})r(\bar{u}) \int \phi\left(y\right)dy = s(\bar{u})r(\bar{u}),
\end{align}
where the second equality uses the normalization of the kernel $\phi(x)$. The outer integral can be evaluated in a similar manner, resulting in the condition,
\begin{align}
\label{equation: homogeneous steady state equation}
m(\bar{u}) = \frac{p_1p_2s(\bar{u})r(\bar{u})}{p_3+p_4s(\bar{u})r(\bar{u})},
\end{align}
which must be satisfied by every homogeneous equilibrium $\bar{u}$. By assumption, $s(0)=0$ and $r(0)=0$, implying that $u=0$ is a homogeneous equilibrium for any parameter setting. Ecologically, this equilibrium corresponds to a bare soil state with no vegetation. Whether (\ref{equation: homogeneous steady state equation}) also admits one or more positive solutions, representing vegetated states, depends on the particular choices of the functions $s(u)$ and $r(u)$. For ecologically realistic choices of these functions, we expect at least one positive homogeneous equilibrium to exist over part of the parameter space. Thus, from this point forward, we let $\bar{u}>0$ denote an arbitrary positive homogeneous equilibrium that satisfies the condition (\ref{equation: homogeneous steady state equation}).

\subsection{Linear stability analysis}
\label{section: linear stability analysis}
We perturb the reduced scalar model (\ref{model: reduced nonlocal scalar equation}) around the positive homogeneous equilibrium $\bar{u}$ by employing the Fourier ansatz,
\begin{align}
\label{equation: Fourier ansatz}
  u(x,t) = \bar{u} + \varepsilon e^{\omega t+i k x} +\rm{c.c.}.
\end{align}
Here, $\varepsilon>0$ is assumed to be small and $\rm{c.c.}$ denotes the complex conjugate. The dispersion relation $\omega=\omega(k;\bar{u})$ represents the growth rate of the perturbation associated with the wavenumber $k$. We substitute this ansatz into (\ref{model: reduced nonlocal scalar equation}) and expand the functions $m(u)$, $s(u)$ and $r(u)$ around $\bar{u}$, which yields,
\begin{align}
\label{equation: first form of the full linearization}
&\omega \cdot \varepsilon e^{\omega t + i k x} = p_1\left(s(\bar{u}) + s'(\bar{u}) \cdot \varepsilon e^{\omega t + i k x}\right)\cdot     \int \phi\left(\frac{x'-x}{r(\bar{u})+r'(\bar{u})\cdot\varepsilon e^{\omega t + i kx}+\mathcal{O}(\varepsilon^2)}\right)\nonumber\\
&\cdot     \,\theta \left(\int \phi\left(\frac{x''-x'}{r(\bar{u})+r'(\bar{u})\cdot\varepsilon e^{\omega t + i kx''}+\mathcal{O}(\varepsilon^2)}\right)\cdot \left(s(\bar{u})+s'(\bar{u})\cdot\varepsilon e^{\omega t + ikx''}\right)dx''\right)dx'\nonumber\\ & -\left(m(\bar{u}) + m'(\bar{u})\cdot\varepsilon e^{\omega t + i k x}\right) - D k^2 \cdot \varepsilon e^{\omega t + i k x} + \mathcal{O}(\varepsilon^2).
\end{align}
To identify all leading-order terms, we first expand $\phi$ and its argument, and then evaluate the inner integral. Next, we expand the function $\theta$. Finally, we evaluate the outer integral.

\subsubsection{The inner integral}

We start by expanding the argument of $\phi$ within the inner integral around $\varepsilon = 0$, and subsequently expand $\phi$ itself around $\frac{x''-x'}{r(\bar{u})}$:
\begin{align}
\label{equation: phi expansion}
&\phi\left(\frac{x''-x'}{r(\bar{u})+\varepsilon e^{\omega t + i kx''}r'(\bar{u})+\mathcal{O}(\varepsilon^2)}\right)\nonumber\\ &\quad =
     \phi\left(\left(\frac{1}{r(\bar{u})}-\varepsilon e^{\omega t + i kx''}\frac{r'(\bar{u})}{r^2(\bar{u})}\right)\left(x''-x'\right)+\mathcal{O}(\varepsilon^2)\right)\nonumber\\
     &\quad = \phi\left(\frac{x''-x'}{r(\bar{u})}\right)-\varepsilon e^{\omega t + i kx''}\frac{r'(\bar{u})}{r^2(\bar{u})}\left(x''-x'\right)\phi'\left(\frac{x''-x'}{r(\bar{u})}\right) +\mathcal{O}(\varepsilon^2).
\end{align}
Using this expansion, the inner integral can be written as the sum
\begin{align}
    \mathcal{I}_1+\varepsilon e^{\omega t + i k x'}\left(\mathcal{I}_2(k)+\mathcal{I}_3(k)\right)+\mathcal{O}(\varepsilon^2),
\end{align}
where the integrals $\mathcal{I}_1$, $\mathcal{I}_2(k)$ and $\mathcal{I}_3(k)$ are expressed and evaluated below.
\begin{enumerate}
\item
The first integral $\mathcal{I}_1$ evaluates to,
\begin{align}
\label{evaluation: I_1}
    \mathcal{I}_1 \coloneqq s(\bar{u})\int \phi\left(\frac{x''-x'}{r(\bar{u})}\right)dx'' = s(\bar{u})r(\bar{u}),
\end{align}
following the calculations in Section \ref{subsection: homogeneous equilibria}.
\item 
For the second integral $\mathcal{I}_2(k)$, we (again) introduce the change of integration variable $y = \frac{x''-x'}{r(\bar{u})}$, and obtain,
\begin{align}
\label{evaluation: I_2}
\begin{split}
    \mathcal{I}_2(k) &\coloneqq e^{-ikx'} s'(\bar{u}) \int \phi\left(\frac{x''-x'}{r(\bar{u})}\right) e^{i k x''} dx''\\ &=  s'(\bar{u})r(\bar{u}) \int \phi(y)e^{i \left(r(\bar{u})k\right) y}\,dy\\
    &= s'(\bar{u})r(\bar{u})\hat{\phi}\left(r(\bar{u}) k\right),
\end{split}
\end{align}
where $\hat{\phi}(k)$ denotes the Fourier transform of $\phi(x)$.
\item
The third integral $\mathcal{I}_3(k)$ can be evaluated with a similar approach:
\begin{align}
\begin{split}
    \mathcal{I}_3(k) &\coloneqq - e^{-ikx'} \frac{s(\bar{u}) r'(\bar{u})}{r^2(\bar{u})} \int (x''-x') \phi'\left(\frac{x''-x'}{r(\bar{u})}\right) e^{i k x''} dx''\\
    &= - s(\bar{u}) r'(\bar{u})\int y \,\phi'(y)e^{i \left(r(\bar{u}) k\right) y}\,dy\\ &= -  s(\bar{u})r'(\bar{u})\hat{\psi}\left(r(\bar{u}) k\right),
\end{split}
\end{align}
where we introduced the function $\psi(x) \coloneqq \phi'(x)x$ through its Fourier transform $\hat{\psi}(k)$.
\end{enumerate}

\subsubsection{\texorpdfstring{The function $\theta$}{The function theta}}

As a next step, we expand $\theta$ around $s(\bar{u})r(\bar{u})$:
\begin{align}
\begin{split}
    &\theta\left(s(\bar{u})r(\bar{u})+\varepsilon e^{\omega t + i k x'}\left(\mathcal{I}_2(k)+\mathcal{I}_3(k)\right)+\mathcal{O}(\varepsilon^2)\right) = \\&\quad\quad\quad\quad\quad\quad\theta\left(s(\bar{u})r(\bar{u})\right)+\varepsilon e^{\omega t + i k x'}\theta'\left(s(\bar{u})r(\bar{u})\right)\left(\mathcal{I}_2(k)+\mathcal{I}_3(k)\right)+\mathcal{O}(\varepsilon^2),
\end{split}
\end{align}
where the derivative of $\theta$ is given by,
\begin{align}
\label{derivative of theta}
    \theta'(y) = -\frac{p_2p_4}{\left(p_3+p_4y\right)^2}.
\end{align}

\subsubsection{The outer integral}

By expanding the kernel function analogously to (\ref{equation: phi expansion}), the outer integral can be written as,
\begin{align}
\begin{split}
    &\int \left(\phi\left(\frac{x'-x}{r(\bar{u})}\right)-\varepsilon e^{\omega t + i kx}\frac{r'(\bar{u})}{r^2(\bar{u})}\left(x'-x\right)\phi'\left(\frac{x'-x}{r(\bar{u})}\right)\right)\\&\quad\quad\cdot\left(\theta\left(s(\bar{u})r(\bar{u})\right)+\varepsilon e^{\omega t + ikx'}\theta'\left(s(\bar{u})r(\bar{u})\right)\left(\mathcal{I}_2(k)+\mathcal{I}_3(k)\right)\right)dx'+ \mathcal{O}\left(\varepsilon^2\right).
\end{split}
\end{align}
As for the inner integral, we decompose the outer integral into a sum of integrals:
\begin{align}
\label{equation: expression outer integral}
    \mathcal{J}_1+\varepsilon e^{\omega t + ikx}\left(\mathcal{J}_2(k)+\mathcal{J}_3(k)\right)+\mathcal{O}\left(\varepsilon^2\right),
\end{align}
and express and evaluate the integrals $\mathcal{J}_1$, $\mathcal{J}_2(k)$ and $\mathcal{J}_3(k)$ below.

\begin{enumerate}
\item
We evaluate $\mathcal{J}_1$ following the calculations in Section \ref{subsection: homogeneous equilibria}:
\begin{align}
    \mathcal{J}_1 \coloneqq \theta\left(s(\bar{u})r(\bar{u})\right) \int \phi\left(\frac{x'-x}{r(\bar{u})}\right)dx' = r(\bar{u}) \, \theta\left(s(\bar{u})r(\bar{u})\right).
\end{align}
\item 
The second integral $\mathcal{J}_2(k)$ can be evaluated similarly to (\ref{evaluation: I_2}),
\begin{align}
\begin{split}
    \mathcal{J}_2(k) &\coloneqq e^{-ikx} \,\theta'\left(s(\bar{u})r(\bar{u})\right)\left(\mathcal{I}_2(k)+\mathcal{I}_3(k)\right) \int \phi\left(\frac{x'-x}{r(\bar{u})}\right) e^{i k x'} dx'\\ &= r(\bar{u})\,\theta'\left(s(\bar{u})r(\bar{u})\right)\left(\mathcal{I}_2(k)+\mathcal{I}_3(k)\right)\int \phi(y)e^{i \left(r(\bar{u})k\right) y}\,dy\\
    &= r(\bar{u})\,\theta'\left(s(\bar{u})r(\bar{u})\right)\left(\mathcal{I}_2(k)+\mathcal{I}_3(k)\right)\hat{\phi}(r(\bar{u})k).
\end{split}
\end{align}
\item
Finally, $\mathcal{J}_3(k)$ evaluates to,
\begin{align}
\begin{split}
    \mathcal{J}_3(k) &\coloneqq - \frac{r'(\bar{u})\,\theta\left(s(\bar{u})r(\bar{u})\right)}{r^2(\bar{u})} \int (x'-x) \phi'\left(\frac{x'-x}{r(\bar{u})}\right) dx'\\
    &= - r'(\bar{u})\, \theta\left(s(\bar{u})r(\bar{u})\right)\int y \,\phi'(y)\,dy\\ &= r'(\bar{u})\, \theta\left(s(\bar{u})r(\bar{u})\right).
\end{split}
\end{align}
The last step uses the following identity,
\begin{align}
\label{identity: useful}
    \int y \,\phi'(y)\,dy = \left[y\phi(y)\right]^\infty_{-\infty}-\int\phi(y)\,dy = - 1,
\end{align}
where we used integration-by-parts together with Lemma \ref{lemma: x phi(x) to zero} in Appendix \ref{appendix: proofs}.
\end{enumerate}

\subsubsection{The leading-order equation}

We substitute the expansion of the outer integral (\ref{equation: expression outer integral}) into the perturbed system (\ref{equation: first form of the full linearization}) and collect all lower order terms:
\begin{align}
\begin{split}
    &\omega\cdot\varepsilon e^{\omega t+ikx} = p_1\Bigg(s(\bar{u})r(\bar{u})\,\theta\left(s(\bar{u})r(\bar{u})\right)+\varepsilon e^{\omega t + ikx}\bigg(\left.\frac{d}{du}\Big(s(u)r(u)\Big)\right|_{u=\bar{u}}\theta\left(s(\bar{u})r(\bar{u})\right)\\&+s(\bar{u})r(\bar{u})\,\theta'(s(\bar{u})r(\bar{u}))\left(s'(\bar{u})r(\bar{u})\hat{\phi}(r(\bar{u})k)-s(\bar{u})r'(\bar{u})\hat{\psi}(r(\bar{u})k)\right)\hat{\phi}(r(\bar{u})k)\bigg)\Bigg)\\&-\left(m(\bar{u})+m'(\bar{u})\cdot\varepsilon e^{\omega t+ikx}\right)-D_uk^2\cdot \varepsilon e^{\omega t+ikx}+\mathcal{O}(\varepsilon^2).
    \end{split}
\end{align}
The $\mathcal{O}(1)$-terms cancel since $\bar{u}$ is a homogeneous equilibrium that satisfies (\ref{equation: homogeneous steady state equation}). Next, we collect all $\mathcal{O}(\varepsilon)$-terms and divide by $\varepsilon e^{\omega t+ikx}$ to obtain a first expression for the dispersion relation $\omega(k;\bar{u})$,
\begin{align}
\begin{split}
\label{equation: dispersion relation first version}
    &\omega(k;\bar{u}) = p_1\left.\frac{d}{du}\Big(s(u)r(u)\Big)\right|_{u=\bar{u}}\theta\left(s(\bar{u})r(\bar{u})\right)-m'(\bar{u})-D_uk^2\\&\quad\quad\quad+p_1s(\bar{u})r(\bar{u})\,\theta'(s(\bar{u})r(\bar{u}))\left(s'(\bar{u})r(\bar{u})\hat{\phi}(r(\bar{u})k)-s(\bar{u})r'(\bar{u})\hat{\psi}(r(\bar{u})k)\right)\hat{\phi}(r(\bar{u})k).
    \end{split}
\end{align}

\subsection{Reformulating the dispersion relation}

As a first step towards simplifying the expression for the dispersion relation, we invoke the homogeneous steady-state equation (\ref{equation: homogeneous steady state equation}) to combine the first two terms of (\ref{equation: dispersion relation first version}) into
\begin{align}
    m(\bar{u})\left(\frac{\left.\frac{d}{du}\Big(s(u)r(u)\Big)\right|_{u=\bar{u}}}{s(\bar{u})r(\bar{u})}-\frac{m'(\bar{u})}{m(\bar{u})}\right) = \frac{m^2(\bar{u})}{s(\bar{u})r(\bar{u})}\left.\frac{d}{du}\left(\frac{s(u)r(u)}{m(u)}\right)\right|_{u=\bar{u}}.
\end{align}
Next, we note that the dispersion relation depends on the Fourier transform of the function $\psi(x) = \phi'(x)x$. To eliminate this dependence, we derive an identity relating $\hat{\psi}(k)$ to $\hat{\phi}(k)$ and its derivative. We begin by differentiating $\hat{\phi}(k)$ with respect to $k$:
\begin{align}
\begin{split}
\label{equation: relation between phi_prime, phi and psi}
    \hat{\phi}'(k) &= \frac{d}{dk}\int\phi(x)e^{ikx}dx = \int\phi(x)\frac{d}{dk}\left(e^{ikx}\right)dx = i \int x\phi(x)e^{ikx}dx  \\ &=  i\left(\left[x\phi(x)\frac{e^{ikx}}{ik}\right]^\infty_{-\infty}-\frac{1}{ik}\int\left(\phi(x)+x\phi'(x)\right)e^{ikx}dx\right)\\ &=-\frac{1}{k}\left(\hat{\phi}(k)+\hat{\psi}(k)\right).
\end{split}
\end{align}
In the last equality, we used the bound $|e^{ikx}| \leq 1$ together with $\lim_{x \rightarrow \pm \infty} x \phi(x)=0$, as established by Lemma \ref{lemma: x phi(x) to zero} in Appendix \ref{appendix: proofs}. Rearranging the terms yields the identity,
\begin{align}
    \hat{\psi}(k) = - \left(\hat{\phi}(k)+k\hat{\phi}'(k)\right),
\end{align}
which we can use to eliminate $\hat{\psi}\left(r(\bar{u}) k\right)$ from (\ref{equation: dispersion relation first version}). Finally, we introduce the function $\beta(k;a)$ as,
\begin{align}
\label{definition: function beta}
    \beta(k;a) \coloneqq \left(1+a\right)\hat{\phi}\left( k\right) + k \hat{\phi}'\left(k\right).
\end{align}
Combining the above simplifications, we can restate the dispersion relation (\ref{equation: dispersion relation first version}) as,
\begin{align}
\begin{split}
\label{equation: dispersion relation final version}
    \omega(k;\bar{u}) = &m(\bar{u})\bigg(\frac{m(\bar{u})}{s(\bar{u})r(\bar{u})}\left.\frac{d}{du}\left(\frac{s(u)r(u)}{m(u)}\right)\right|_{u=\bar{u}}\\&\quad\quad\quad\quad\quad\quad\quad\quad\quad\quad\quad-\frac{p_4s(\bar{u})r'(\bar{u})}{p_3+p_4s(\bar{u})r(\bar{u})}\hat{\phi}(r(\bar{u})k) \beta(r(\bar{u})k;a)\bigg)-D_uk^2,
    \end{split}
\end{align}
where,
\begin{align}
\label{definition: a}
    a = \frac{s'(\bar{u})r(\bar{u})}{s(\bar{u})r'(\bar{u})} = \frac{\frac{s'(\bar{u})}{s(\bar{u})}}{\frac{r'(\bar{u})}{r(\bar{u})}}>0.
\end{align}

\section{Conditions for vegetation pattern formation}
\label{section: pattern formation conditions}

In this section, we determine under which conditions the scalar model (\ref{model: reduced nonlocal scalar equation}) admits vegetation pattern formation through a Turing instability. For a homogeneous equilibrium $\bar{u}$ to be Turing unstable, it must be temporally stable (i.e., stable against homogeneous perturbations) and spatially unstable (i.e., unstable against heterogeneous perturbations). Furthermore, the system must suppress the appearance of unstable perturbations with arbitrarily large wavenumbers, indicating well-posedness. We verify these properties by analyzing the sign of the dispersion relation $\omega(k;\bar{u})$ as given in (\ref{equation: dispersion relation final version}). Specifically, the equilibrium $\bar{u}$ is temporally stable when $\omega(0;\bar{u})<0$, while it is spatially unstable when $\omega(k;\bar{u})>0$ over some range of non-zero wavenumbers $k$. The system is not ill-posed when $\lim_{k \rightarrow \pm\infty} \omega(k;u)<0$.

\subsection{Well-posedness and temporal stability}

\subsubsection{Well-posedness}
\label{section: well-posedness}

Because the non-negative kernel $\phi(x)$ is normalized, we have $\lim_{k \rightarrow \pm\infty} \hat{\phi}(k) = 0$ (by the Riemann-Lebesgue lemma). Moreover, $\lim_{k \rightarrow \pm\infty} \beta(k;a) = 0$ (see Lemma \ref{lemma: beta(k) to zero} in Appendix \ref{appendix: proofs}). Together, these limits imply
\begin{align}
    \lim_{k \rightarrow \pm \infty} \omega(k;u) = -\infty,
\end{align}
with the stabilizing effect of biomass spread modelled as diffusion dominating for large wavenumbers. Hence, exponentially growing perturbations with arbitrarily large wavenumbers are suppressed, making the model (\ref{model: reduced nonlocal scalar equation}) not ill-posed.

\subsubsection{Temporal stability}

The existence of a positive homogeneous equilibrium in the scalar model (\ref{model: reduced nonlocal scalar equation}) depends on the specific choices of the functions $s(u)$ and $r(u)$, as well as the parameter values (see Section \ref{subsection: homogeneous equilibria}). However, whenever such an equilibrium exists, and all homogeneous equilibria are non-degenerate (i.e., not located at a bifurcation point), it can be shown that there must exist a positive homogeneous equilibrium $\bar{u}$ that is temporally stable (see Appendix \ref{appendix: temporally stable}).

\subsection{Spatial instabilities}

The dispersion relation in (\ref{equation: dispersion relation final version}) consists of three contributions: a constant term, a term generated by non-local water uptake and a term that comes from biomass spread modelled by diffusion. The influence of the diffusive contribution, $-D_u k^2$, is the most straightforward: it decreases the value of the dispersion relation for every $k \neq 0$, and its effect becomes stronger when $|k|$ increases. Consequently, biomass spread alone cannot induce a spatial instability, since such an instability requires $\omega(k;\bar{u})$ to become positive over some range of non-zero wavenumbers $k$. On the contrary, biomass spread may even prevent a spatial instability by counteracting the (potential) destabilizing effects of the remaining two terms. This is in clear contrast with the role of biomass spread in classical vegetation-water models \citep{Klausmeier1999,HilleRisLambers2001,Gilad2004}, where the interplay between fast water diffusion and slow biomass spread drives pattern formation through a diffusion-driven Turing instability.

Positive values of the dispersion relation must therefore originate from one or both of the remaining terms. Thus, for $\bar{u}$ to be spatially unstable, either the constant term must be positive,
\begin{align}
\label{necessary condition patterning I}
\left.\frac{d}{du}\left(\frac{s(u)r(u)}{m(u)}\right)\right|_{u=\bar{u}}>0,
\end{align}
or the second term must take positive values for at least some wavenumbers $k$, which can be stated equivalently as
\begin{align}
\label{necessary condition patterning II}
\inf_{k} \left(\hat{\phi}(k) \beta (k;a)\right) < 0,
\end{align}
since $s(\bar{u}),r(\bar{u}),r'(\bar{u})>0$ (see Section \ref{section: model equations}).

Below, we analyze when (\ref{necessary condition patterning I}) and (\ref{necessary condition patterning II}) become sufficient to drive pattern formation through a Turing instability. Furthermore, we provide an ecological interpretation of the underlying mechanisms.

\subsubsection{Pattern formation driven by root growth outpacing mortality}

First, we examine the pattern-forming condition (\ref{necessary condition patterning I}), which requires the constant term of the dispersion relation to be positive. As shown in Section \ref{section: well-posedness}, the contribution to the dispersion relation that arises from non-local water uptake vanishes as $k$ diverges. Hence, for sufficiently large $k$, say $k>\tilde{k}$, the dispersion relation is mainly governed by the constant term and the diffusive contribution. When $D_u$ is small enough, diffusion is too weak to overcome the positive constant term for wavenumbers $k$ immediately beyond $\tilde{k}$. As a result, $\omega(k,\bar{u})>0$ over a non-empty interval of wavenumbers. Thus, the condition (\ref{necessary condition patterning I}) generates a Turing instability when biomass spread is sufficiently slow. This instability mechanism is based on the same underlying principle as the growth-outpacing-susceptibility mechanism introduced in \citep{vanderVoort2026}. 

To assess the ecological interpretation of the condition (\ref{necessary condition patterning I}) for pattern formation, we first consider the root distribution function
\begin{align}
    \chi(x,x',t) = \phi\left(\frac{x'-x}{r(u(x,t))}\right)s(u(x,t)),
\end{align}
which represents the root density at the location $x'$ originating from above-ground biomass at $x$. Integrating this distribution over space gives the total below-ground biomass $b(x,t)$ of the root system associated with the above-ground biomass at $x$:
\begin{align}
\begin{split}
    b(u(x,t)) &= \int \phi\left(\frac{x'-x}{r(u(x,t))}\right)s(u(x,t))dx'\\ &= s(u(x,t))r(u(x,t))\int \phi\left(y\right)dy \\&= s(u(x,t))r(u(x,t)).
\end{split}
\end{align}
Thus, the product $s(u)r(u)$ equals the below-ground (root) biomass $b(u)$ associated with an above-ground biomass density $u$. The condition (\ref{necessary condition patterning I}) therefore requires the ratio $b(u)/m(u)$ to be increasing at the homogeneous equilibrium $\bar{u}$. In other words, root growth must increase faster than mortality as the above-ground biomass increases.

The double integral arising from an explicit description of overlapping root systems does not affect the existence criteria of this instability mechanism. In contrast, root allometry promotes its occurrence. Specifically, strong root allometry corresponds to a larger value of $r'(u)$, causing the root biomass $b(u) = s(u)r(u)$ to increase more rapidly with the above-ground biomass density, potentially causing root growth to outpace mortality. Nonetheless, root allometry is not necessary for this instability mechanism: even when $r(u)$ is constant, pattern formation can still occur when the total root density $s(u)$ increases faster than the vegetation mortality $m(u)$ as the above-ground biomass density increases.

\begin{figure}[t!]
    \centering
    \begin{minipage}[b]{0.49\textwidth}
        \centering
        \includegraphics[width=\textwidth]{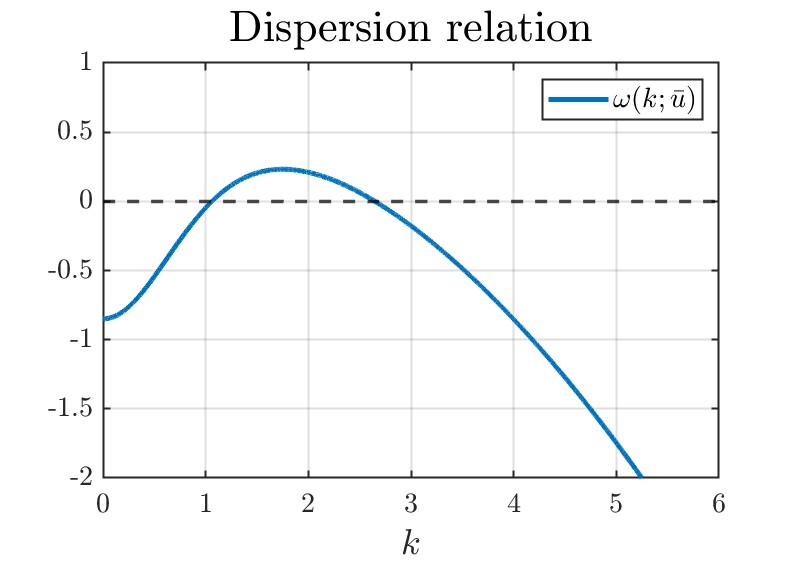}
    \end{minipage}%
        \hspace{0.1cm}
    \begin{minipage}[b]{0.49\textwidth}
        \centering
        \includegraphics[width=\textwidth]{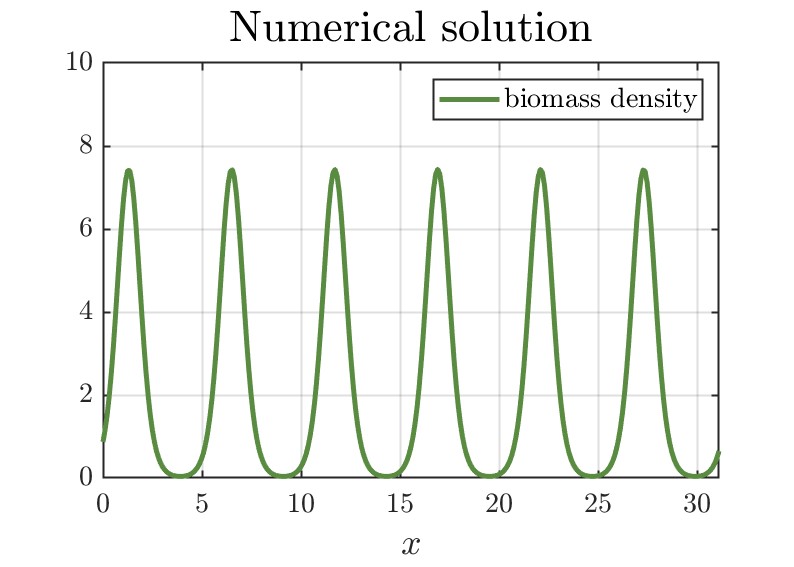}
    \end{minipage}%
\caption{Pattern formation induced by root growth outpacing mortality. The left panel shows that the dispersion relation takes an excursion above zero, indicating a Turing instability. The right panel shows the corresponding vegetation pattern obtained from a numerical simulation.}
\label{figure: spatial patterns root to shoot}
\end{figure}

\begin{example}
In the model (\ref{model: reduced nonlocal scalar equation}), we take $m(u)$, $s(u)$ and $r(u)$ linear, i.e., $m(u) = q_mu$, $s(u)=q_su$ and $r(u)=\frac{q_{r_1}u}{1+q_{r_2}u}$ for some constants $q_m,q_s,q_{r_1},q_{r_2}>0$. For this choice, a positive temporally stable homogeneous equilibrium $\bar{u}$ exists whenever
\begin{align}
    p_2>\frac{q_m}{p_1q_sq_{r_1}}\left(2\sqrt{p_3p_4q_sq_{r_1}}+p_3q_{r_2}\right).
\end{align}
Moreover, 
\begin{align}
    \frac{d}{du}\left(\frac{s(u)r(u)}{m(u)}\right)=\frac{q_sq_{r_1}}{q_m}\frac{1}{\left(1+q_{r_2}u\right)^2}>0,
\end{align}
so, root growth consistently outpaces mortality. The dispersion relation evaluated at $\bar{u}$ reads,
\begin{align}
    \omega(k;\bar{u}) = \frac{q_m}{1+q_{r_2}\bar{u}}\left(1-\frac{p_4q_m}{p_1p_2}\bar{u}\,\hat{\phi}\left(\frac{q_{r_1}\bar{u}}{1+q_{r_2}\bar{u}}k\right)\beta\left(\frac{q_{r_1}\bar{u}}{1+q_{r_2}\bar{u}}k;1+q_{r_2}\bar{u}\right)\right) -D_uk^2.
\end{align}
The left panel of Figure \ref{figure: spatial patterns root to shoot} shows $\omega(k;\bar{u})$ as a function of $k$ for a Laplace kernel with $\mu = 0.19$ (see Appendix \ref{appendix: Laplace kernel}), and the parameters $p_1 = 2.0$, $p_2 = 165$, $p_3 = 140$, $p_4 = 140$, $q_m = 1.0$, $q_s = 1.4$, $q_{r_1} = 1.9$, $q_{r_2} = 0.15$ and $D_u = 0.10$. Because the dispersion relation is positive over a non-empty interval of wavenumbers, the homogeneous equilibrium $\bar{u}$ is Turing unstable. The right panel of Figure \ref{figure: spatial patterns root to shoot} shows the corresponding numerical simulation of (\ref{model: reduced nonlocal scalar equation}), and confirms that the Turing instability gives rise to stable vegetation patterns.
\end{example}

\subsubsection{Pattern formation through pressure zones driven by root allometry}
\label{section: pattern formation due to pressure zones driven by root allometry}

Next, we examine the pattern-forming condition (\ref{necessary condition patterning II}). This condition resembles a well-known criterion for pattern formation through a Turing instability in scalar models with non-local resource competition represented by a single convolution term \citep{Fuentes2004,Segal2013,Martinez2013a,Martinez2013b,Martinez2014,Tega2022}, namely
\begin{align}
\label{condition: necessary mechanism single integral}
\inf_{k} \hat{\phi}(k) < 0.
\end{align}
This condition states that the Fourier transform of the kernel function must take negative values for at least some wavenumbers. This cannot occur when $\phi(x)$ is positive definite (by Bochner's theorem). For instance, this means that the Gaussian and Laplace kernels do not satisfy (\ref{condition: necessary mechanism single integral}). Nonetheless, the assumptions imposed on $\phi(x)$ in Section \ref{section: model equations} (namely that it is positive, symmetric, normalized and non-increasing in $|x|$) do not require positive definiteness. Indeed, many ecologically relevant kernel functions, including the parabolic, cosine and top-hat kernels, do satisfy (\ref{condition: necessary mechanism single integral}) as their Fourier transforms take negative values \citep{vanderVoort2026}. Such negative excursions are more likely for kernels that decay relatively slowly near their symmetry point, but more quickly further away \citep{Leimar2008,Pigolotti2010,Barabas2012,Leimar2013}.

The ecological interpretation of the condition (\ref{condition: necessary mechanism single integral}) is based on the formation of competitive pressure zones\footnote{We use the term ``competitive pressure zones" rather than the commonly used ``exclusion areas" following the discussion in \citep{vanderVoort2026}. After a supercritical Turing bifurcation, regions between high-biomass areas experience increased competitive pressure rather than complete exclusion of vegetation.} \citep{Pigolotti2007,Pigolotti2010,Martinez2013a,Martinez2013b,vanderVoort2026}. Kernels with moderately flat centers and rapidly decaying tails can generate a higher cumulative competitive pressure in the regions between high-biomass areas than within these areas themselves. This spatial asymmetry reinforces biomass accumulation inside denser patches, while vegetation growth is suppressed in the surrounding regions, thereby promoting the emergence of vegetation patterns. In some cases (specifically, when the Turing bifurcation is subcritical \citep{vanderVoort2026}), this instability mechanism may immediately lead to exclusion areas in between densely vegetated areas, corresponding to bare soil regions where vegetation cannot establish \citep{Pigolotti2007,Pigolotti2010,Martinez2013a,Martinez2013b}.

To test whether competitive pressure zones also form in the reduced scalar model (\ref{model: reduced nonlocal scalar equation}), we performed a series of numerical experiments. We selected parameter sets for which the condition (\ref{necessary condition patterning I}) was not met, to ensure that observed Turing instabilities were generated solely by (\ref{necessary condition patterning II}). For each simulation, we evaluated the inner integral of (\ref{model: reduced nonlocal scalar equation}) (which quantifies the spatial distribution of competitive pressure, as discussed in Section \ref{section: double integral}) together with the resulting biomass distribution after vegetation patterns had formed. Across all considered choices of the functions $s(u)$ and $r(u)$, the competitive pressure was consistently higher in the regions between high-biomass areas than within these areas themselves (see the right panel of Figure \ref{figure: spatial patterns pressure zones} for a representative example simulation). These numerical observations support the idea that the formation of competitive pressure zones also underlies the instability mechanism associated with the condition (\ref{necessary condition patterning II}).

In contrast to the instability criterion (\ref{condition: necessary mechanism single integral}) for scalar models with a single convolution term, the condition (\ref{necessary condition patterning II}) may also be satisfied by kernels whose Fourier transform is non-negative for all wavenumbers. In Appendix \ref{appendix: kernels}, we show this explicitly for both the Gaussian and the Laplace kernel. For such kernels, the influence of $a$ (as defined in (\ref{definition: a})) on the occurrence of a Turing instability becomes particularly clear. To show this, we first note that it directly follows from the definition of $\beta(k;a)$ in (\ref{definition: function beta}) that
\begin{align}
\frac{d}{d a} \beta(k;a) = \hat{\phi}(k).
\end{align}
Hence, for every fixed wavenumber $k$, the function $\beta(k;a)$ increases with $a$ for kernels with a non-negative Fourier transform. Consequently, decreasing $a$ lowers the minimum of $\beta(k;a)$ (whenever it exists). Since the condition (\ref{necessary condition patterning II}) requires this minimum to be negative when $\hat{\phi}(k)\geq0$, smaller values of $a$ promote the occurrence of a Turing instability.

The numerator and the denominator of $a$ represent the per capita growth rates of $s(u)$ and $r(u)$, respectively, evaluated at the homogeneous equilibrium $\bar{u}$. So, $a$ measures how additional biomass is allocated between increasing the total root density and extending the lateral reach of the root system. Small values of $a$ therefore result in sparse root systems that spread out over larger distances, while larger values lead to denser root systems that are more localized. Consequently, our analysis indicates that a Turing instability is more likely when additional below-ground biomass is invested primarily in lateral root extension than in increasing the total root density. This is consistent with the general principle that long-range competitive interactions promote spatial pattern formation, whereas purely local competition does not. 

However, for kernels for which the Fourier transform does take negative values, the reasoning above no longer applies since $\beta(k;a)$ is no longer monotonically increasing in $a$. In this scenario, the impact of $a$ on whether or not a Turing instability occurs is no longer as straightforward.

In any case, root allometry plays an important role in enabling pattern formation through condition (\ref{necessary condition patterning II}), since this instability criterion depends explicitly on $r(u)$ through $a$. To emphasize this, we consider the absence of root allometry by assuming that the lateral root extent is constant, i.e., $r(u) \coloneqq 1$. Then, the linearization of the double integral in (\ref{model: reduced nonlocal scalar equation}) becomes a self-convolution term similar to the MacArthur-Levins overlap kernel in niche theory \citep{Macarthur1967,Schoener1974}. In those models, the homogeneous equilibrium is spatially stable when uptake and depletion are proportional, and therefore, in particular, when they are governed by the same kernel function \citep{HernandezGarcia2009}. 

For our model, the dispersion relation (\ref{equation: dispersion relation first version}) for $r(u) \coloneqq 1 $ reduces to
\begin{align}
\begin{split}
\label{equation: dispersion relation no root allometry}
    \omega(k;\bar{u}) = \,&m(\bar{u})\left(\frac{m(\bar{u})}{s(\bar{u})}\left.\frac{d}{du}\left(\frac{s(u)}{m(u)}\right)\right|_{u=\bar{u}}-\frac{p_4s'(\bar{u})}{p_3+p_4s(\bar{u})}\hat{\phi}^2(k)\right)-D_uk^2.
    \end{split}
\end{align}
The contribution that comes from non-local water uptake now contains the non-negative square $\hat{\phi}^2(k)$ instead of the product $\hat{\phi}(r(\bar{u})k) \beta (r(\bar{u})k;a)$ which may take negative values. Hence, the pressure zone instability mechanism disappears, and pattern formation through this pathway is no longer possible. The appearance of the squared Fourier transform (as opposed to a single Fourier transform) is a direct consequence of the double non-local integral in (\ref{model: reduced nonlocal scalar equation}), which arises from explicitly accounting for root overlap. Thus, root-to-root competition suppresses the formation of competitive pressure zones, analogously to the role of overlapping kernels in niche theory \citep{HernandezGarcia2009}.

Notably, root allometry breaks the proportionality between uptake and depletion, since the root extents for vegetation at $x$ and $x''$ within the double integral in (\ref{model: reduced nonlocal scalar equation}) depend on their respective biomass densities $u(x,t)$ and $u(x'',t)$. Consequently, uptake and depletion are governed by different kernel functions. More specifically, root allometry counteracts the stabilizing effect of overlap by replacing one of the Fourier transforms with the function $\beta(k;a)$. This allows the contribution to the dispersion relation that comes from non-local water uptake to become positive, thereby restoring the possibility of pattern formation through competitive pressure zones. Thus, root allometry is an essential ingredient for this instability mechanism.

Finally, we note that the condition (\ref{necessary condition patterning II}) alone is not sufficient for pattern formation. The contribution to the dispersion relation that arises from non-local water uptake must also dominate both the diffusive contribution and the constant term (which we assume to be non-positive to isolate the pressure zone instability mechanism) for the dispersion relation to become positive. As before, the stabilizing effect of diffusion can be made arbitrarily small by taking $D_u$ sufficiently small. Under this assumption, the following condition is sufficient for the occurrence of a Turing instability:
\begin{align}
\label{condition: necessary mechanism IIa}
\inf_{k} \left(\hat{\phi}(k) \beta(k;a)\right) < \frac{m(\bar{u})}{s(\bar{u})r(\bar{u})}\left.\frac{d}{du}\left(\frac{s(u)r(u)}{m(u)}\right)\right|_{u=\bar{u}}\frac{p_3+p_4s(\bar{u})r(\bar{u})}{p_4s(\bar{u})r'(\bar{u})}.
\end{align}
This inequality serves as a practical criterion for determining whether competitive pressure zones can induce pattern formation for a particular choice of the root functions $s(u)$ and $r(u)$ in the reduced scalar model (\ref{model: reduced nonlocal scalar equation}).

\begin{figure}[t!]
    \centering
    \begin{minipage}[b]{0.52\textwidth}
        \centering
        \includegraphics[width=\textwidth]{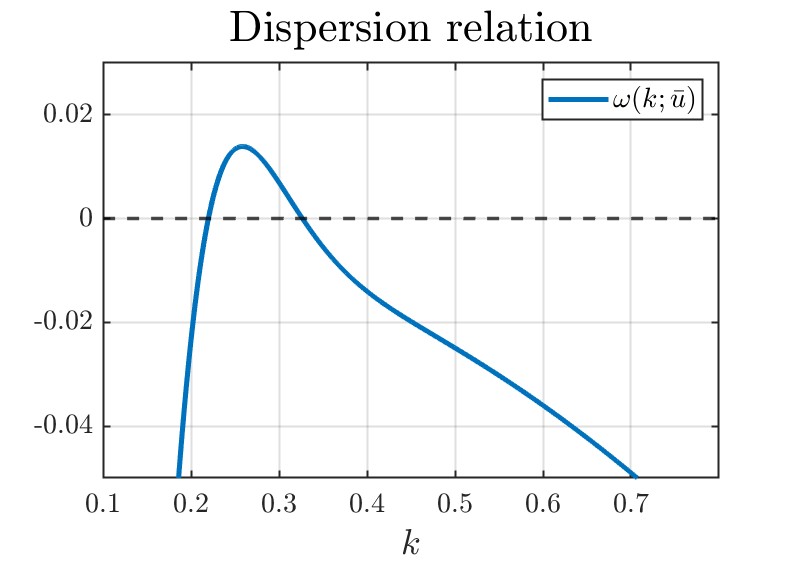}
    \end{minipage}%
        \hspace{-0.7cm}
    \begin{minipage}[b]{0.52\textwidth}
        \centering
        \includegraphics[width=\textwidth]{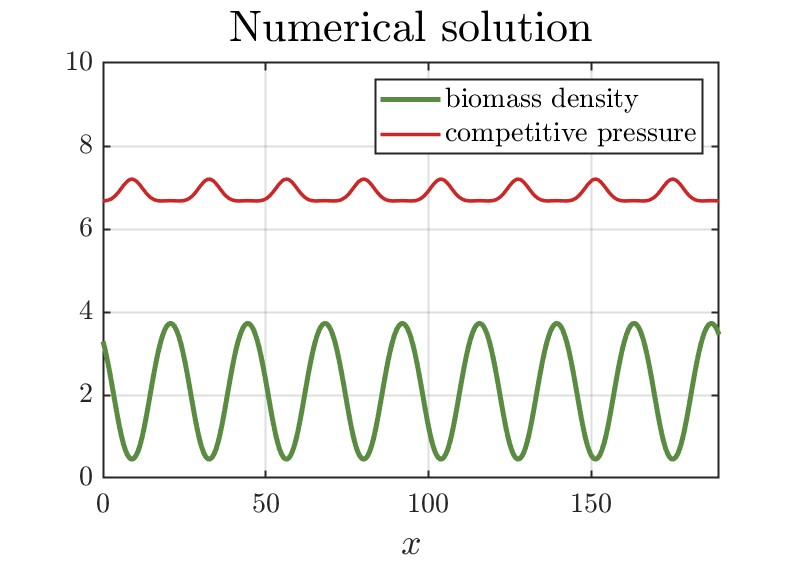}
    \end{minipage}%
\caption{Pattern formation through competitive pressure zones induced by root allometry. The left panel shows that the dispersion relation is positive over a non-empty range of wavenumbers, indicating a Turing instability. The right panel shows the corresponding numerical solution together with the competitive pressure. The two quantities are out of phase, with the competitive pressure highest in between the high-biomass areas.}
\label{figure: spatial patterns pressure zones}
\end{figure}

\begin{example}
We choose $m(u) = q_m u$, $s(u)=q_su^\frac{1}{2}$ and $r(u)=q_r u^{\frac{1}{2}}$ in the model (\ref{model: reduced nonlocal scalar equation}). Then, 
\begin{align}
    \bar{u} = \frac{1}{p_4}\left(\frac{p_1p_2}{q_m}-\frac{p_3}{q_sq_r}\right),
\end{align}
is a positive temporally stable homogeneous equilibrium whenever $p_2>\frac{q_mp_3}{p_1q_sq_r}$. Moreover, the condition (\ref{condition: necessary mechanism IIa}) reduces to
\begin{align}
    \inf_k\hat{\phi}(k) \beta(k;1)<0.
\end{align}
Whenever this condition is satisfied and $D_u$ is sufficiently small, a Turing instability is present. The dispersion relation evaluated at $\bar{u}$ reads,
\begin{align}
\begin{split}
    \omega(k;\bar{u}) = &\,-\frac{q_m^2p_4}{2p_1p_2}\bar{u}\hat{\phi}(q_r\bar{u}^\frac{1}{2} k) \beta(q_r\bar{u}^\frac{1}{2}k;1)-D_uk^2
    \end{split}
\end{align}
For a Gaussian kernel with $\mu = 1.9$ (see Appendix \ref{appendix: Gaussian kernel}), and the parameters $p_1 = 1.8$, $p_2 = 184$, $p_3 = 180$, $p_4 = 130$, $q_m = 1.0$, $q_s = 1.7$, $q_{r} = 1.9$ and $D_u = 0.10$, we show in the left panel of Figure \ref{figure: spatial patterns pressure zones} that the dispersion relation takes positive values for a non-empty range of wavenumbers. The right panel shows the resulting vegetation pattern together with the corresponding competitive pressure, illustrating that the regions in between the high-biomass areas experience the highest competitive pressure.
\end{example}

\section{Discussion}

In this study, we investigated how root overlap and root allometry affect vegetation pattern formation. We first formulated the vegetation-water system (\ref{model: two component model with non local water uptake}) which accounts for root-mediated water uptake, including root allometry through the function $r(u)$. We eliminated the water variable in the limit of fast water dynamics, resulting in the reduced scalar vegetation model (\ref{model: reduced nonlocal scalar equation}). The explicit description of soil water depletion in the two-component setting naturally gives rise to a double integral structure in this reduced model, which accounts for below-ground resource competition when root systems overlap, i.e., root-to-root competition. 

For this scalar model, we derived the dispersion relation which we analyzed to determine the conditions under which Turing instabilities can occur. This analysis revealed two distinct instability mechanisms that can lead to the emergence of spatial vegetation patterns. The first mechanism requires root growth to outpace mortality as the above-ground biomass density increases \citep{vanderVoort2026}, while the second mechanism relies on the emergence of competitive pressure zones, in which vegetation growth is suppressed by the cumulative non-local water uptake of surrounding areas with higher biomass density \citep{Pigolotti2007,Pigolotti2010,Martinez2013a,Martinez2013b,vanderVoort2026}.

Root allometry and root overlap affect these two instability mechanisms differently. Regarding the root-growth-outpacing-mortality mechanism, we note that a root system that extends laterally with the above-ground biomass density naturally boosts root growth. Hence, strong root allometry may potentially cause root growth to outpace mortality, and can therefore drive vegetation pattern formation through this instability mechanism. Nonetheless, even when root allometry is fully absent, the instability mechanism may still operate. Root-to-root competition, on the other hand, does not directly influence the existence criterion (\ref{necessary condition patterning I}) of the root-growth-outpacing-mortality mechanism.

Pattern formation through competitive pressure zones, on the other hand, is promoted by root allometry but inhibited by root-to-root competition. In scalar vegetation models neglecting both root allometry and root-to-root competition, non-local resource competition is typically represented through a single convolution. The resulting dispersion relation contains the Fourier transform of the interaction kernel, and a Turing instability can occur when this Fourier transform takes negative values. Accounting for root overlap, however, introduces a double integral into the scalar vegetation model, which results in a squared Fourier transform in the dispersion relation. Since this quantity is non-negative, the pressure zone instability mechanism disappears, similar to the spatial stability observed in niche theory when overlap structures are considered \citep{HernandezGarcia2009}. When root allometry is additionally included, one of the Fourier transforms is replaced by the function $\beta(k;a)$, as defined in (\ref{definition: function beta}), thereby restoring the possibility of pattern formation driven by competitive pressure zones. In particular, for kernels with a non-negative Fourier transform (e.g., kernels with relatively flat centers), competitive pressure zones are more likely to appear when plants invest more resources in lateral root extension than in increasing their total root density, which is consistent with the general idea that long-range competition promotes pattern formation \citep{Martinez2014}.

Our results show that root allometry and root overlap strongly influence the conditions for vegetation pattern formation. Both features should therefore be considered when developing vegetation models for water-limited environments. In several vegetation-water systems, root allometry and root overlap are explicitly represented \citep{Gilad2007,Bennett2023}, but this is typically not the case in models in which soil water is modelled implicitly \citep{Lefever1997,Martinez2013b,vanderVoort2026}. Our reduction procedure reveals a way to incorporate both features into such scalar models. Root allometry can be represented by allowing the spatial extent of the interaction kernel to increase with the above-ground biomass density, similar to previous work \citep{Lefever2009,Cisternas2020,Tlidi2024}. Root-to-root competition requires a double integral structure (as in (\ref{model: reduced nonlocal scalar equation})) to model the impact of root overlap accurately, instead of the commonly adopted single convolution representation. Such a double integral may also be relevant for the description of any situation in which organisms forage over overlapping areas, such as central place foragers like ants and termites \citep{Maciel2021}.

Throughout this work, we focused on vegetation pattern formation driven by the ability of plants to access water beyond their above-ground location through their laterally extended root systems. To isolate the associated instability mechanisms, we restricted our attention to ecosystems in which soil water diffusion is negligible, thereby excluding diffusion-driven instabilities. Future work should extend our results by introducing soil water diffusion into (\ref{model: two component model with non local water uptake}) to investigate how the instability mechanisms identified here interact with diffusion-driven instabilities. In particular, it would be interesting to determine whether these mechanisms reinforce or suppress one another. 

In our model set-up, we did not consider some of the facilitative interactions that can occur in arid ecosystems. Nonetheless, facilitation plays an important role in vegetation pattern formation; for example, increased water infiltration in vegetated patches can generate spatial self-organization \citep{HilleRisLambers2001,Rietkerk2002}. Similarly, reduced evaporation due to shading has been shown to influence both pattern formation and pattern dynamics \citep{vonHardenberg2001,Gilad2007,vanderVoort2025}. Building on our results, future efforts could investigate how local facilitative interactions, such as those occurring near the stem, and non-local facilitative interactions, for instance within the radius of the crown, interact with root allometry and root overlap to influence the conditions for vegetation pattern formation.

Another natural extension of our framework would be to consider multiple competing plant species. With such an extension, one could investigate whether differences in root allometry characteristics between species can enable species coexistence through spatiotemporal niche formation \citep{Eigentler2020,Maciel2021}.

Lastly, we mainly focused on the emergence of vegetation patterns and did not investigate the dynamics of the resulting patterns far from the pattern-forming bifurcation point. Understanding how spatial vegetation patterns respond to gradual or sudden changes in environmental conditions could provide further insights into ecosystem resilience \citep{Siteur2014,Bastiaansen2020}. Recent work has started to address this question by showing that the shape of non-local competition and facilitation kernels affects both pattern onset and pattern resilience \citep{Clifton2026}. A natural next step is to investigate whether patterns induced by the different instability mechanisms identified here respond differently to environmental change, and, in particular, whether one mechanism can maintain spatial patterns under conditions in which the other fails.

To summarize, our work addresses the influence of root overlap and root allometry on vegetation pattern formation. Our results show that both features should be considered in studies of spatial vegetation patterns in water-stressed ecosystems, as they alter the existence criteria of the underlying instability mechanisms.

\subsubsection*{Acknowledgments}
This research was supported by the Dutch Research Council (NWO) through the project ``Resilience in Complex Systems through Adaptive Spatial Pattern Formation'' (project number OCENW.M20.169), with contributions from J.V. and A.D. Additional funding for A.D. was provided by the European Research Council through the ERC-Synergy project RESILIENCE (proposal number 101071417). This work was partially funded by the Center of Advanced Systems Understanding (CASUS), which is financed by Germany’s Federal Ministry of Research, Technology and Space (BMFTR) and by the Saxon Ministry for Science, Culture and Tourism (SMWK) with tax funds on the basis of the budget approved by the Saxon State Parliament. RMG was also supported by the São Paulo Research Foundation (FAPESP) through grant ICTP-SAIFR 2021/14335-0.

\begin{appendices}

\section{Balance between total water uptake and depletion}
\label{appendix: balance}

In this appendix, we show that the non-local interactions in the vegetation-water system (\ref{model: two component model with non local water uptake}) preserve the water balance of the system. More precisely, we demonstrate that for every time $t$, the total water uptake by vegetation over the spatial domain equals the total depletion of soil water. We begin by integrating the water uptake term in the vegetation equation over space:
\begin{align}
    \int s(u(x,t))\int\phi\left(\frac{x'-x}{r(u(x,t))}\right)w(x',t)\,dx'dx
    \\= \int\int s(u(x,t))\phi\left(\frac{x'-x}{r(u(x,t))}\right)w(x',t)\,dx'dx\\
    \overset{(1)}{=} \int\int s(u(x,t))\phi\left(\frac{x'-x}{r(u(x,t))}\right)w(x',t)\,dxdx'\\
    \overset{(2)}{=} \int\int s(u(x',t))\phi\left(\frac{x-x'}{r(u(x',t))}\right)w(x,t)\,dx'dx\\
    \overset{(3)}{=} \int\int s(u(x',t))\phi\left(\frac{x'-x}{r(u(x',t))}\right)w(x,t)\,dx'dx\\
    = \int w(x,t)\int \phi\left(\frac{x'-x}{r(u(x',t))}\right)s(u(x',t))\,dx'dx.
\end{align}
In equality (1), we use Fubini's theorem to interchange the order of integration. In equality (2), the integration variables are relabelled ($x \rightarrow x'$ and $x' \rightarrow x$). Finally, equality (3) follows from the symmetry property of the kernel function $\phi(x)$.

The final expression equals the integral over space of the water depletion term of the water equation. This confirms the consistency of the proposed non-local formulation (\ref{model: two component model with non local water uptake}).

\section{Properties of the kernel function and its Fourier transform}
\label{appendix: proofs}

\begin{lemma}
\label{lemma: x phi(x) to zero}
    Let $\phi(x)$ be non-negative and non-increasing as a function of $|x|$. Furthermore, assume that
    \begin{align}
        \int \phi(x) \,dx < \infty,
    \end{align}
    holds. Then,
    \begin{align}
        \lim_{x \rightarrow \pm \infty} x \phi(x) = 0.
    \end{align}
\end{lemma}

\begin{proof}
    We assume $\lim_{x \rightarrow \infty} x \phi(x) \neq 0$. Then, there exists an $\varepsilon > 0$ and a non-negative sequence $\{x_n\}_{n \in \mathbb{N}}$ with $\lim_{n \rightarrow \infty} x_n = \infty$, such that $x_n \phi(x_n) \geq \varepsilon$ for all $n \in \mathbb{N}$.

    We consider an infinite subsequence $\{y_n\}_{n \in \mathbb{N}}$ of $\{x_n\}_{n \in \mathbb{N}}$ such that the intervals $[\frac{y_n}{2},y_n]$ are disjoint, which is possible because $\lim_{n \rightarrow \infty} x_n = \infty$. Then,
    \begin{align}
        \infty &> \int^\infty_0 \phi(x) dx \geq \sum^\infty_{n=0}\int^{y_n}_{\frac{y_n}{2}}\phi(x) dx \geq \sum^\infty_{n=0}\left(y_n-\frac{y_n}{2}\right)\phi(y_n)\geq \frac{1}{2}\sum^\infty_{n=0}\varepsilon = \infty,
    \end{align}
    which is clearly a contradiction. So, we must have $\lim_{x \rightarrow \infty} x\phi(x) = 0$.

    Proving $\lim_{x \rightarrow -\infty} x\phi(x) = 0$ goes along the same lines.
\end{proof}

\begin{lemma}
\label{lemma: beta(k) to zero}
    Let $\phi(x)$ satisfy the assumptions of Lemma \ref{lemma: x phi(x) to zero}. Then,
    \begin{align}
        \lim_{k \rightarrow \pm \infty} \beta (k;a) = 0.
    \end{align}
\end{lemma}

\begin{proof}
    We recall from (\ref{equation: relation between phi_prime, phi and psi}) that
    \begin{align}
        \hat{\phi}'(k) = -\frac{1}{k}\left(\hat{\phi}(k)+\int x\phi'(x) e^{ikx}dx\right),
    \end{align}
    which implies
    \begin{align}
        \beta(k;a) = a\hat{\phi}(k)-\int x\phi'(x) e^{ikx}dx.
    \end{align}
    The first term converges to $0$ as $k \rightarrow \pm\infty$ by the Riemann-Lebesgue lemma. The same argument can be applied to the second term, since
    \begin{align}
    \int |x \,\phi'(x)| dx = -\int x \,\phi'(x) dx = 1 < \infty,
    \end{align}   
    where we used the identity (\ref{identity: useful}) in the last equality.
\end{proof}

\section{The existence of a temporally stable homogeneous equilibrium}
\label{appendix: temporally stable}

\begin{lemma}
\label{lemma: temporal stability}
If the model (\ref{model: reduced nonlocal scalar equation}) admits a positive homogeneous equilibrium $u^*$, and all homogeneous equilibria are non-degenerate, then there exists a positive temporally stable homogeneous equilibrium $\bar{u}$.
\end{lemma}

\begin{proof}
    Let $u^*$ be an arbitrary positive homogeneous equilibrium of (\ref{model: reduced nonlocal scalar equation}). Then, by (\ref{equation: homogeneous steady state equation}), it satisfies
    \begin{align}
    \label{equation: another version of the hsse}
        p_1s(u)r(u)\theta(s(u)r(u))=m(u).
    \end{align}
    We refer to the left- and right-hand side of this equation as $LHS(u)$ and $RHS(u)$, respectively. By the non-degeneracy of $u^*$, we have $LHS'(u^*)\neq RHS'(u^*)$. First, we consider $LHS'(u^*) < RHS'(u^*)$. We can express the dispersion relation $\omega(k;u)$ evaluated at $k=0$ as follows 
    \begin{align}
    \label{equation: dispersion relation at k=0}
    \begin{split}
        \omega(0;u) = LHS'(u)-RHS'(u),
    \end{split}
    \end{align}
    directly implying that $\omega(0;u^*) < 0$. So, in this scenario $u^*$ itself is temporally stable, and we can take $\bar{u} = u^*$.
    
    Next, we consider the case $LHS'(u^*) > RHS'(u^*)$. This inequality together with the continuity of the curves $LHS(u)$ and $RHS(u)$ for $u>0$ implies that $LHS(u)>RHS(u)$ for all $u>u^*$ sufficiently close to $u^*$. Furthermore, by substituting the definition of $\theta$ from (\ref{definition: theta}) into (\ref{equation: another version of the hsse}), we find that the left-hand side is confined in $[0,\frac{p_1p_2}{p_4}]$ for all $u \geq 0$ while the right-hand side converges to $\infty$ for $u \rightarrow \infty$. This yields that $LHS(u) < RHS(u)$ for sufficiently large $u$. Hence, by the intermediate value theorem, there exists at least one non-degenerate homogeneous equilibrium larger than $u^*$. Let $\bar{u}>u^*$ denote the largest non-degenerate homogeneous equilibrium, which exists as $LHS(u)<RHS(u)$ for sufficiently large $u$. If $LHS'(\bar{u})>RHS'(\bar{u})$, we would have $LHS(u)> RHS(u)$ for all $u>\bar{u}$ which contradicts that $LHS(u)<RHS(u)$ for sufficiently large $u$. Thus, $LHS'(\bar{u})<RHS'(\bar{u})$.

    Finally, by (\ref{equation: dispersion relation at k=0}), this inequality implies that $\omega(0;\bar{u})<0$. Therefore, $\bar{u}$ is temporally stable.
\end{proof}

\section{Kernel functions and their properties}
\label{appendix: kernels}

\subsection{Laplace kernel}
\label{appendix: Laplace kernel}
The Laplace kernel function can be written as,
\begin{align}
    \phi(x) = \frac{1}{2\mu}e^{-\frac{|x|}{\mu}},
\end{align}
where the parameter $\mu>0$ is the mean interaction distance. Its Fourier transform equals,
\begin{align}
    \hat{\phi}(k) = \frac{1}{\left(\mu k\right)^2+1},
\end{align}
which is strictly positive for all $k$. Therefore, $\hat{\phi}(k) \beta(k;a)$ becomes negative exactly when $\beta(k;a)$ becomes negative. This function can be calculated to be,
\begin{align}
    \beta(k;a) &= \frac{\left(a-1\right)\left( \mu k\right)^2 + a + 1}{\left(\left( \mu k\right)^2+1\right)^2}.
\end{align}
It follows directly that $\beta(k;a)$ takes negative values if and only if $a < 1$. In this case, the function becomes negative for $|k|$ large enough, in particular, when,
\begin{align}
    |k| > \frac{1}{\mu}\sqrt{\frac{1+a}{1-a}}.
\end{align}

\subsection{Gaussian kernel}
\label{appendix: Gaussian kernel}
The Gaussian kernel function can be written as,
\begin{align}
    \phi(x) = \frac{1}{\pi \mu}e^{-\frac{1}{\pi}\frac{x^2}{\mu^2}},
\end{align}
where the parameter $\mu>0$ is the mean interaction distance. Its Fourier transform equals,
\begin{align}
    \hat{\phi}(k) = e^{-\frac{\pi}{4}\left(\mu k\right)^2},
\end{align}
which is strictly positive for all $k$. Therefore, $\hat{\phi}(k) \beta(k;a)$ becomes negative exactly when $\beta(k;a)$ becomes negative. This function can be calculated to be,
\begin{align}
    \beta(k;a) &= \left(1+a-\frac{\pi}{2}\left(\mu k\right)^2\right)e^{-\frac{\pi}{4}\left(\mu k\right)^2}.
\end{align}
It follows directly that for every $a>0$, $\beta(k;a)$ becomes negative for $|k|$ large enough, in particular, when,
\begin{align}
    |k| > \frac{1}{\mu}\sqrt{\frac{2\left(1+a\right)}{\pi}}.
\end{align}

\end{appendices}

\bibliographystyle{unsrtnat} 
\bibliography{ref}

\end{document}